%% file: main-arxiv.tex
\documentclass{article}
\usepackage{kr}

\usepackage{times}
\usepackage{soul}
\usepackage{url}
\usepackage[hidelinks]{hyperref}
\usepackage[utf8]{inputenc}
\usepackage[small]{caption}
\usepackage{graphicx}
\usepackage{amsmath}
\usepackage{amsthm}
\usepackage{booktabs}
\usepackage{algorithm}
\usepackage{algpseudocode} 
\usepackage{mathtools}
\usepackage{enumerate}
\usepackage[bibliography=common]{apxproof} 

\usepackage{bbding}
\usepackage{amssymb}
\usepackage{xspace}
\input{kl-config}
\input{kl-complexity}

\input{kl-abbreviations}
\input{kl-knowledges}
\input{cmd.tex}

\newtheorem{theorem}{Theorem}
\newtheorem{lemma}[theorem]{Lemma}
\newtheoremrep{lemma}[theorem]{Lemma}

\newtheorem{observation}[theorem]{Observation}
\newtheoremrep{corollary}[theorem]{Corollary}
\newtheoremrep{proposition}[theorem]{Proposition}

\newtheorem{claim}[theorem]{Claim}
\newtheoremrep{claim}[theorem]{Claim}

\newtheorem{open}[theorem]{Open problem}

\crefname{thm}{Theorem}{Theorems}

\title{Boundedness for Unions of Conjunctive Regular Path Queries \\over Simple Regular Expressions}

\author{%
Diego Figueira$^1$\and
S. Krishna$^2$\and
Om Swostik Mishra$^2$\and
Anantha Padmanabha$^3$
\affiliations
$^1$Univ. Bordeaux, CNRS,  Bordeaux INP, LaBRI, UMR 5800, Talence, France\\
$^2$Indian Institute of Technology Bombay, Mumbai, India\\
$^3$Indian Institute of Technology Madras, Chennai, India\\
\emails
diego.figueira@cnrs.fr,
\{krishnas, 21b090022\}@iitb.ac.in,
ananthap@cse.iitm.ac.in
}

\begin{document}

\maketitle

\begin{abstract}

\input{abstract}
\end{abstract}

\input{knowledge-notice}

\section{Introduction}
\label{sec-Intro}

\input{intro}

\section{Preliminaries}
\label{sec-prelims}
\input{prelims}

\section{Main Results}
\label{sec-Results}
\input{results}

\section{Succinct Automata and Succinct CQs}
\label{sec-Automata}
\input{automata}

\input{ssec-cont-succCQ}

\section{Upper Bound}
\label{sec-upperBound}
\input{upper}

\section{Lower Bound}
\label{sec-lowerBound}
\input{lower-arxiv}

\input{sec-A-boundedness-arxiv}

\section{Conclusion}
\label{sec-conclusion}
\input{conclusion}
\section*{Acknowledgements}
  Diego Figueira is partially supported by ANR AI Chair INTENDED, grant ANR-19-CHIA-0014.

\bibliographystyle{kr}
\bibliography{long,ref}

\end{document}

%% file: kl-config.tex
\usepackage[xcolor, hyperref, cleveref, notion, quotation, electronic]{knowledge}
\usepackage{mathcommand}
\knowledgeconfigure{quotation, protect quotation={tikzcd}}
\knowledgeconfigure{diagnose line=true, diagnose bar=true}

\definecolor{Dark Ruby Red}{HTML}{7c1b1e}
\definecolor{Dark Blue Sapphire}{HTML}{00343d} 
\definecolor{Dark Gamboge}{HTML}{be7c00}

\IfKnowledgePaperModeTF{
}{
    \knowledgestyle{intro notion}{color={Dark Ruby Red}, emphasize}
    \knowledgestyle{notion}{color={Dark Blue Sapphire}}
    \hypersetup{
        colorlinks=true,
        breaklinks=true,
        linkcolor={Dark Blue Sapphire}, 
        citecolor={Dark Blue Sapphire}, 
        filecolor={Dark Blue Sapphire}, 
        urlcolor={Dark Blue Sapphire},
    }
    \IfKnowledgeElectronicModeTF{
    }{
        \knowledgeconfigure{anchor point color={Dark Ruby Red}, anchor point shape=corner}
        \knowledgestyle{intro unknown}{color={Dark Gamboge}, emphasize}
        \knowledgestyle{intro unknown cont}{color={Dark Gamboge}, emphasize}
        \knowledgestyle{kl unknown}{color={Dark Gamboge}}
        \knowledgestyle{kl unknown cont}{color={Dark Gamboge}}
    }
}

%% file: kl-complexity.tex
\knowledge{wrap=\textsf}
  | NL

\knowledge{notion, text={paraNL}, wrap=\textsf}
  | para-NL
  | paraNL

\knowledge{notion, wrap=\textsf}
  | FPT
  | fixed-parameter tractable

\knowledge{text={ExpSpace}, wrap=\textsf}
  | EXPSPACE
  | ExpSpace

\knowledge{text={2ExpSpace}, wrap=\textsf}
  | 2EXPSPACE
  | 2ExpSpace

\knowledge{text={PSpace}, wrap=\textsf}
  | PSpace
  | PSPACE

\knowledge{text={FP}, wrap=\textsf}
  | FP

\knowledge{text={NP}, wrap=\textsf}
  | NP

  \knowledge{text={co-NP}, wrap=\textsf}
  | coNP

  \knowledge{text={P}, wrap=\textsf}
  | P

  \knowledge{text={\ensuremath{\# \textsf{P}}}, wrap=\textsf}
  | shP

  \knowledge{text={\ensuremath{\textsf{P}^{\# \textsf{P}}}}, wrap=\textsf}
  | PshP
  
\knowledge{text={\ensuremath{\textsf{FP}^{\# \textsf{P}}}}, wrap=\textsf}
  | FPshP

\knowledge{text={\ensuremath{\Sigma^p_2}}, wrap=\textsf}
  | SigmaP2

\knowledge{text={\ensuremath{\Pi^p_2}}, wrap=\textsf}
  | PiP2
  | Pi2
  | pi2
  | Pitwo
  | PiTwo
  | pitwo

\knowledge{url={https://en.wikipedia.org/wiki/Savitch\%27s_theorem}}
  | Savitch's Theorem

%% file: kl-abbreviations.tex
\knowledge{text={i.e.}, italic}
  | ie

\knowledge{text={I.e.}, italic}
  | Ie

\knowledge{text={s.t.}, italic}
  | st

\knowledge{text={e.g.}, italic}
  | eg

\knowledge{text={vs.}, italic}
  | vs

\knowledge{text={w.r.t.}, italic}
  | wrt

\knowledge{text={a.k.a.}, italic}
  | aka

\knowledge{text={w.l.o.g.}, italic}
  | wlog

\knowledge{text={cf.}, italic}
  | cf

\knowledge{text={iff}, italic}
  | iff

\knowledge{text={r.e.}, italic}
  | r.e.
  | re

%% file: kl-knowledges.tex
\knowledge{notion}
 | notion@notice
 | definition@notice

\knowledge{notion}
 | graph database
 | Graph Databases
 | Graph database
 | graph databases
 | database
 
\knowledge{notion}
 | oriented path
 | oriented paths
 | Oriented Path
 | Oriented path
 
\knowledge{notion}
 | label
 | Label

 \knowledge{notion}
 | expansion
 | expansions
 | Expansion

\knowledge{notion}
 | Models
 | models
 
\knowledge{notion}
 | regular path query
 | Regular Path Query
 | Regular path query
 | RPQ
 | RPQs
 
\knowledge{notion}
 | 2RPQ
 | 2RPQs
 
\knowledge{notion}
 | Conjunctive 2RPQ
 | Conjunctive 2RPQs
 | C2RPQ
 | C2RPQs
 
\knowledge{notion}
 | Conjunctive RPQ
 | Conjunctive RPQs
 | CRPQ
 | CRPQs
 
\knowledge{notion}
 | union of C2RPQ
 | union of C2RPQs
 | UC2RPQ
 | UC2RPQs
 
 \knowledge{notion}
 | union of CRPQ
 | union of CRPQs
 | UCRPQ
 | UCRPQs
 | (U)CRPQs
 
\knowledge{notion}
 | conjunctive query
 | CQ
 | CQs
 
\knowledge{notion}
 | unions of CQs
 | UCQ
 | UCQs
 | Unions of CQs
 | (U)CQs

\knowledge{notion}
 | equivalent
 
\knowledge{notion}
 | bounded
 | unbounded
 | Bounded
 | boundedness
 
\knowledge{notion}
 | factor
 
\knowledge{notion}
 | succinct-NFA

\knowledge{notion}
 | initial state

\knowledge{notion}
 | final states
 | final state

\knowledge{notion}
 | succinct-transitions
 | succinct-transition

\knowledge{notion}
 | accepting run
 | accepting runs

\knowledge{notion}
 | membership problem

\knowledge{notion}
 | non-deterministic finite automata
 | NFA

\knowledge{notion}
 | Parikh-image

\knowledge{notion}
 | succinct CQ
 | succinct CQs

\knowledge{notion}
 | containment problem

\knowledge{notion}
 | atom
 | atoms

\knowledge{notion}
 | equality atoms
 | equality atom

\knowledge{notion}
 | $w$-expansion
 | $t$-expansion
 | $n$-expansion
 | $n'$-expansion
 | $3$-expansion
 | $1$-expansion

\knowledge{notion}
 | $m$-expansion
 | $k$-expansion
 | $k$-expansions
 | $s$-expansion
 | $\STEP$-expansion
 | $m$-expansions

 \knowledge{notion}
 | Minimal expansion
 | Minimal expansions 
 | minimal expansion

\knowledge{notion}
 | atom expansion
 | atom expansions
 | expansion path
 | expansion paths
 | expanded path

 \knowledge{notion}
 | Minimal path
 | Minimal paths
 | minimal path
 | minimal paths
 | $(h,\alpha )$-minimal path
 | $(h,\alpha )$-minimal paths
 | $(t,\alpha )$-minimal path
 | $(t,\alpha )$-minimal paths

 \knowledge{notion}
 | mapped point
 | mapped points

 \knowledge{notion}
 | atom $m$-expansion

\knowledge{notion}
 | homomorphism
 | homomorphic
 | homomorphically
 | homomorphisms

 \knowledge{notion}
 | purely-red

 \knowledge{notion}
 | Boolean

\knowledge{notion}
 | size
 | sizes

\knowledge{notion}
 | size@SSF
\knowledge{notion}
 | path
 | paths

 \knowledge{notion}
 | non-recursive
 | recursive atoms
 | recursive
 | recursive atom

 \knowledge{notion}
 | interval
 | intervals

\knowledge{notion}
 | contracting
 | contract
 | contraction

\knowledge{text={\IsBdd\ problem}}
 | boundedness problem

\knowledge{notion}
 | star-free regular expressions

\knowledge{notion}
 | succinct star-free expressions

\knowledge{notion}
 | isomorphic copy
 
\knowledge{notion}
 | $(A,m)$-expansion
 | $(A,m)$-expansions

\knowledge{notion}
 | contained
 | containment

\knowledge{notion}
 | free variables

\knowledge{notion}
 | satisfies

\knowledge{notion}
 | corresponding

%% file: cmd.tex
\usepackage{soul}

\renewcommand{\epsilon}{\varepsilon}

\newcommand{\Nat}{\ensuremath{\mathbb{N}}}
\newcommand{\set}[1]{\ensuremath{\{#1\}}}

\definecolor{light-gray}{gray}{0.9}
\newcommand{\proofcase}[1]{\noindent\colorbox{light-gray}{#1}~}

\newcommand{\np}{\textup{\sc NP}\xspace}

\newcommand{\piPTwo}{\ensuremath{\Pi^p_2}\xspace}

\newrobustcmd{\defeq}{\mathrel{\hat{=}}}

\knowledgenewrobustcmd{\alphabet}{\cmdkl{\ensuremath{\mathbb{A}}}}
\knowledgenewrobustcmd{\reg}{\cmdkl{\textup{RE}}}

\knowledgenewrobustcmd{\aSingleton}{\cmdkl{\ensuremath{\mathrm{\bf a}}}}
\knowledgenewrobustcmd{\aStar}[1][]{\cmdkl{\ensuremath{\mathrm{\bf a}_{#1}^*}}}
\knowledgenewrobustcmd{\wSingleton}{\cmdkl{\ensuremath{\mathrm{\bf w}}}}
\knowledgenewrobustcmd{\wSingletonN}{\cmdkl{\ensuremath{\mathrm{\bf w}^n}}}
\knowledgenewrobustcmd{\wStar}{\cmdkl{\ensuremath{\mathrm{\bf w}^*}}}
\knowledgenewrobustcmd{\WStar}{\cmdkl{\ensuremath{\mathrm{\bf W}^*}}}

\knowledgenewrobustcmd{\SF}{\cmdkl{\textup{SF}\xspace}}
\knowledgenewrobustcmd{\SSF}{\cmdkl{\textup{SSF}\xspace}}

\newcommand{\class}{\ensuremath{\mathcal{Q}}\xspace} 
\knowledgenewrobustcmd{\classOp}[1]{\ensuremath{\cmdkl{\mathcal{Q}(}#1\cmdkl{)}}\xspace} 

\knowledgenewrobustcmd{\vars}{\cmdkl{\textit{vars}}} 
\knowledgenewrobustcmd{\Exp}{\cmdkl{\ensuremath{\textup{Exp}}}} 
\knowledgenewrobustcmd{\Nmax}{\cmdkl{\ensuremath{\textup{N}}}} 
\knowledgenewrobustcmd{\Atom}{\cmdkl{\ensuremath{\textup{Atom}}}} 
\knowledgenewrobustcmd{\STEP}{\cmdkl{\ensuremath{\Pi}}} 
\knowledgenewrobustcmd{\Min}{\cmdkl{\ensuremath{\textup{Min}}}} 
\knowledgenewrobustcmd{\model}{\mathrel{\cmdkl{\models}}} 
\knowledgenewrobustcmd{\notmodel}{\mathrel{\cmdkl{\not\models}}} 
\knowledgenewrobustcmd{\pref}{\cmdkl{\ensuremath{\textup{pref}}}} 
\knowledgenewrobustcmd{\suff}{\cmdkl{\ensuremath{\textup{suff}}}} 
\knowledgenewrobustcmd{\E}{\cmdkl{\ensuremath{\textup{E}}}} 
\knowledgenewrobustcmd{\V}{\cmdkl{\ensuremath{\textup{V}}}} 
\newrobustcmd{\collapse}{\approx}
\knowledgenewrobustcmd{\homto}{\mathrel{\cmdkl{\xrightarrow{\smash{\textit{\tiny hom}}}}}}

\knowledgenewrobustcmd{\semsubset}{\mathrel{\cmdkl{\subseteq}}} 
\knowledgenewrobustcmd{\notsemsubset}{\mathrel{\cmdkl{\not\subseteqq}}} 
\knowledgenewrobustcmd{\semequiv}{\mathrel{\cmdkl{\equiv}}} 

\knowledgenewrobustcmd{\Expm}[2][q]{\cmdkl{\textup{Exp}_{#2}\!(}#1\cmdkl{)}} 
\knowledgenewrobustcmd{\ExpAm}[2][q]{\cmdkl{\textup{Exp}_{#2}\!(}#1\cmdkl{)}} 
\knowledgenewrobustcmd{\qm}[2][q]{\cmdkl{#1(}#2\cmdkl{)}} 
\knowledgenewrobustcmd{\qAm}[1]{\cmdkl{q(}#1\cmdkl{)}} 

\knowledgenewrobustcmd{\IsBdd}{\cmdkl{\sc{Boundedness}\xspace}}
\knowledgenewrobustcmd{\bddIn}[1]{\cmdkl{#1\text{-bounded}}}
\knowledgenewrobustcmd{\bddInPb}[1]{\cmdkl{#1\textsc{-boundedness}}}

\knowledgenewrobustcmd{\AutomataLang}[1]{\cmdkl{\mathcal{L}(#1)}} %
\knowledgenewrobustcmd{\factor}[2]{\cmdkl{[#1 .. #2]}} %
\newcommand{\automata}{\mathcal{A}}

\knowledgenewrobustcmd{\allHom}{\cmdkl{\ensuremath{\mathcal{H}}}}
\knowledgenewrobustcmd{\lambdas}[1][s]{\cmdkl{\lambda}^{\cmdkl{(}#1\cmdkl{)}}}
\knowledgenewrobustcmd{\sq}{\cmdkl{s_q}}

\knowledgenewrobustcmd{\nratoms}[2][]{\cmdkl{|}#2\cmdkl{|}^{#1}_{\cmdkl{\textsf{at}}}} 
\knowledgenewrobustcmd{\nrvars}[2][]{\cmdkl{|}#2\cmdkl{|}^{#1}_{\cmdkl{\textsf{var}}}} 

\knowledgenewrobustcmd{\Zbound}{\cmdkl{Z_q}}
\knowledgenewrobustcmd{\Zboundplus}{\cmdkl{Z_q^+}}
\knowledgenewrobustcmd{\ZredBound}{\cmdkl{Z_{\textit{red}}}}
\knowledgenewrobustcmd{\Nbound}{\cmdkl{N_q}}
\knowledgenewrobustcmd{\Recwords}{\cmdkl{\mathcal{R}_q}}
\knowledgenewrobustcmd{\ZcolBound}{\cmdkl{Z_{\textit{col}}}}

\knowledgenewrobustcmd{\qA}{\cmdkl{q(A)}}

\newcommand{\bigQL}{\ucrpq{\SSF,\wStar}}
\newcommand{\ucrpq}[1]{\textup{"UCRPQ"(\ensuremath{#1})}\xspace}
\newcommand{\crpq}[1]{\textup{"CRPQ"(\ensuremath{#1})}\xspace}

%% file: abstract.tex
The problem of checking whether a recursive query can be rewritten as query without recursion is a fundamental reasoning task, known as the boundedness problem.
Here we study the boundedness problem for Unions of Conjunctive Regular Path Queries (UCRPQs), a navigational query language extensively used in ontology and graph database querying. 

The boundedness problem for UCRPQs is "ExpSpace"-complete. 
Here we focus our analysis on UCRPQs using simple regular expressions, which are of high practical relevance and enjoy a lower reasoning complexity.

We show that the complexity for the boundedness problem for this UCRPQs fragment is "PiP2"-complete, and that an equivalent bounded query can be produced in polynomial time whenever possible.

When the query turns out to be unbounded, we also study the task of finding an  equivalent maximally bounded query, which we show to be feasible in "PiP2".
As a side result of independent interest stemming from our developments, we study a notion of succinct finite automata and prove that its membership problem is in "NP".

%% file: knowledge-notice.tex

\noindent
\raisebox{-.4ex}{\HandRight}\ \ This pdf contains internal links: clicking on a "notion@@notice" leads to its \AP ""definition@@notice"".\footnote{\url{https://ctan.org/pkg/knowledge}}

%% file: intro.tex

    Regular Path Queries (RPQ) and its extension under conjunctions, known as Conjunctive RPQ (CRPQ) are a well-known generalization of conjunctive queries with a mild form of recursion, which are extensively used for querying knowledge bases and graph-structured datasets.
    In particular, CRPQs are part of SPARQL, which is the W3C standard for querying RDF data, including well known knowledge bases such as DBpedia and Wikidata. In particular, RPQs are extensively used according to recent studies \cite{BonifatiMT19,MalyshevKGGB18}. More generally, CRPQs are basic building blocks for querying graph databases \cite{AnglesABHRV17,Baeza13}. 

    As knowledge bases become larger, reasoning about queries (e.g. for optimization) becomes increasingly important. 
    In this vein, many static analysis aspects of CRPQ have been studied, starting with the seminal papers on containment of CRPQs \cite{FlorescuLS98,CalvaneseGLV00} ("ie", whether the results of a query $q_1$ always contain those of $q_2$), which spawned many subsequent works.
   In particular, for queries with recursion such as CRPQ, a basic reasoning task is whether the recursion can be bounded. In other words, whether a given query can be equivalently written as a finite union of conjunctive queries (UCQ), known as the "boundedness problem".
    UCQs form the core of most systems for data management and ontological query answering, and, in addition, are the focus of advanced optimization techniques.
    The "boundedness problem" has garnered attention, in particular for Datalog programs. For ontology-mediated query answering (OMQA), the problem is known as the ``FO-rewritability'' or ``UCQ-rewritability'' problem, which typically takes, as ontology-mediated query, a conjunctive query and some TBox formulated in some description logic (see, "eg", \cite{Bienvenu0LW16}). In this sense, the current work can be seen as a preliminary study for investigating FO-rewritability of ontology-mediated CRPQ queries.

    Boundedness of CRPQs has been shown to be decidable, "ExpSpace"-complete \cite{barcelo2019boundedness}, that is, as hard (or easy) as the containment problem for CRPQs \cite{CalvaneseGLV00}. 
    This holds also for the extension with union (UCRPQ) and two-way navigation (UC2RPQ).
    Further, whenever a query is bounded the equivalent UCQ may be of triply exponential size, and hence not directly amenable to an optimization procedure.

    However, the lower bounds for boundedness ---and generally for most static analysis problems on CRPQ--- use rather complex regular expressions that are hardly used in practice. This has raised the question of the status of these problems when restricted to simpler basic regular expressions.

    In fact, recent studies reveal that most CRPQ use extremely simple regular expressions of the form $a^*$ or $a_1 + \dotsb + a_n$. Studies on query-logs \cite{BonifatiMT19,MalyshevKGGB18} show that these expressions cover more than 98\% of all RPQ queries made on Wikidata and more than 46\% of the queries made in DBpedia ---see \cite[Table 1]{FigueiraGKMNT20} for more details.

    A line of research was then established into understanding the difficulty of treating CRPQ over such simple regular expressions. It has been shown that over simpler expressions the situation improves drastically for the containment problem \cite[Table 2]{FigueiraGKMNT20} and the ``semantic tree-width problem'', that is,  whether the query is equivalent to a query of low tree-width \cite[\S 6]{FigueiraM23}.

    \paragraph{Contributions}

        We study the boundedness problem for UCRPQs restricted to a class of simple regular expressions. Such regular expressions can be either of the form $w^*$, where $w$ is a word, or any regular expression which does not use Kleene star. We even allow for ``succinct exponentiation'' of the form $w^n$ and $w^{\leq n}$ where $n$ is encoded in binary.

        For such queries we show that the boundedness problem is "PiP2"-complete. Both the upper and lower-bound proofs are non-trivial, but our main technical contribution is the upper bound (\Cref{sec-upperBound}).
        Whenever the UCRPQ is bounded, it can be written as a UCRPQ of polynomial size which does not contain expressions with Kleene star.

        As a necessary ingredient for the  "PiP2" upper bound proof, in \Cref{sec-Automata} we introduce a notion of ``succinct automata'' and use its membership checking problem, which we prove to be in "NP", to solve the boundedness problem. Succinct automata are classic finite automata extended with transitions of the form $p\xrightarrow{w^n}p'$ to indicate a path of transitions from $p$ to $p'$ reading the word $w^n$, where $n$ is represented in binary. These automata could be of independent interest when looking for succinct models of computation.  
        
    We also consider the related problem of bounding the query ``by letters'', "ie" whether the query is equivalent to one not making any recursion on a given subset of letters (\Cref{sec:letter-boundedness}). In this setting, the boundedness problem corresponds to bounding \emph{all} letters occurring in the query. 
    We show that there is a notion of ``maximally bounded'' query, "ie", a query that is bounded in a maximum number of letters which is also computable within the "PiP2" bound.

Finally, in \Cref{sec-lowerBound} we prove that the "PiP2" lower bound holds even for CRPQs with very simple regular expressions of the form $a^*$ or $a$.

\paragraph{Organization}
 \Cref{sec-prelims} describes the preliminaries needed for the paper. Then we formally state  our main results in \Cref{sec-Results}. We dedicate \Cref{sec-Automata} to the discussed side result on ``succinct automata'' and \Cref{sec-upperBound} to the main upper bound result. \Cref{sec-lowerBound} elaborates on  the lower bound.
\Cref{sec:letter-boundedness} delves into the problem of bounding a query ``by letter''. We conclude with \Cref{sec-conclusion}.

%% file: prelims.tex
\AP
Let $\intro*\alphabet$ be a finite alphabet. A ""graph database"" over $\alphabet$ is a finite edge-labelled directed graph $G = (V, E)$ over $\alphabet$, where $V$ is a finite set  of vertices and $E \subseteq V \times \alphabet \times V$ is the set of labelled edges. We write $u \xrightarrow{a} v$ to denote an edge $(u, a, v) \in E$. 
\knowledgenewrobustcmd{\Image}{\cmdkl{\textit{Im}}}
We write \AP$\intro*\Image(f)$ to denote the image of any function $f$.

\AP
We denote $\alphabet^*$ as the set of all finite words over $\alphabet$ and $\epsilon$ to denote the empty word. A ""path"" from $u$ to $v$ 
in a graph database $G=(V,E)$ over alphabet $\alphabet$ 
is a (possibly empty) sequence of edges $\sigma=v_0 \xrightarrow{a_1} v_1,
v_1 \xrightarrow{a_2}v_2,\ldots, v_{k-1} \xrightarrow{a_k}v_k$ where each $v_i \xrightarrow{a_{i+1}}v_{i+1}\in E$, $v_0=u, v_k=v$. The ""label"" of the path $\sigma$ from $u$ to $v$ is the word $a_1a_2 \dots a_k$ of edge labels seen along the path. When $k=0$, the "label" of the "path" is $\epsilon$; this is called the ``empty "path"'', and there is always an empty "path" from $u$ to $u$, for every vertex $u$.

\AP
Unless otherwise stated, we assume that regular languages  $L \subseteq \alphabet^*$ are encoded via regular expressions.
A ""regular path query"" (\reintro{RPQ}) over $\alphabet$ is a regular language $L$ given as a regular expression. The semantics of $L$ on a "graph database" $G = (V, E)$ over $\alphabet$, written $L(G)$, is the set of pairs $(u, v) \in V \times V$ such that there is a directed "path" from $u$ to $v$ in $G$ whose "label" belongs to $L$.

\AP
A ""Conjunctive RPQ""  (\reintro{CRPQ}) over $\alphabet$ is an expression $q:=\exists \bar{x} \big((y_1\xrightarrow{L_1}z_1)\land (y_2\xrightarrow{L_2}z_2)\land \ldots \land (y_m\xrightarrow{L_m}z_m) \big)$ where each $L_i$ is an "RPQ" over $\alphabet$. 
We call each $y_i \xrightarrow{L_i}z_i$ an ""atom"". 
Further, $\bar{x}$ is a tuple of variables contained in $\{y_1, z_1, \dots, y_m, z_m\}$ and the variables of $q$ not contained in $\bar x$ are the ""free variables"" of $q$. A query with no "free variables" is called ""Boolean"". 

To simplify the definitions and technical developments we shall henceforth assume that all queries we deal with are "Boolean". However, all our results and bounds hold for non-"Boolean" queries as well (modulo some slightly more cluttered definitions and proofs).

\AP
In the context of "graph databases" a ""conjunctive query"" (\reintro{CQ}) over $\alphabet$ can be defined as a "CRPQ" over $\alphabet$ of the form $\exists \bar{z} \bigwedge_{1\le i \le m} x_i\xrightarrow{L_i}y_i$ where each $L_i$ is a regular expression of the form $a$ for some letter $a \in \alphabet$, which denotes the singleton set $\set{a}$.

\AP
Given "Boolean" "CQ"s $q, q'$, a ""homomorphism"" from $q$ to $q'$ is a mapping $h:\vars(q)\rightarrow \vars(q')$ such that for every "atom" $x\xrightarrow{a}y$ of $q$, we have that $h(x)\xrightarrow{a}h(y)$ is an "atom" of $q'$. We write $q \intro*\homto q'$ when such "homomorphism" exists and $h:q \homto q'$ to denote that $h$ is one such "homomorphism". 
A \reintro{homomorphism} $q \homto G$ on a "graph database" $G=(V,E)$ is defined analogously, as a mapping $h:\vars(q)\rightarrow V$ such that for every "atom" $x\xrightarrow{a}y$ of $q$, we have $h(x) \xrightarrow{a} h(y) \in E$. 
    
\paragraph{Expansions}
Any "UCRPQ" can be equivalently seen as an infinite union of "CQs". We define formally the shape of such "CQs", which we call "expansions".

\AP
"CRPQs" with ""equality atoms"" are queries of the form $q(\bar{x}) = \delta \land I$, 
where $\delta$ is a "CRPQ" (without "equality atoms") and $I$ is a conjunction of "equality atoms" of the form $x=y$. 
We define the binary relation $=_q$ over $\vars(q)$ to be the reflexive-symmetric-transitive closure of the binary relation $\{(x, y) \mid \text{$x=y$ is an "equality atom" in $q$}\}$. 
In other words, we have $x=_q y$ if the equality $x=y$ is forced by the "equality atoms" of $q$. 
Note that every "CRPQ" with "equality atoms" $q(\bar{x}) = \delta \land I$ is equivalent to a "CRPQ" without "equality atoms"  $q^{\collapse}$, 
which is obtained from $q$ by collapsing each equivalence class of the relation $=_q$ into a single variable. 
This transformation gives us a \emph{canonical} renaming from $\vars(q)$ to $\vars(q^{\collapse})$. For instance, $q \defeq \exists x,y,z ~~ x \xrightarrow{K} y \land y \xrightarrow{L} z \land (x = y)$
collapses to $q^{\collapse} \defeq \exists x,z ~~ x \xrightarrow{K} x \land x \xrightarrow{L} z$.

\AP
For any "atom" $x \xrightarrow{L} y$ of a "CRPQ" $q$ and $w=a_1a_2 \dots a_k \in L$, the ""$w$-expansion"" of $x \xrightarrow{L} y$ is the "CQ" $P$ defined as follows:
    (i) If $w\neq \epsilon$ then $P \defeq x \xrightarrow{a_1} z_1 \land z_1 \xrightarrow{a_2} z_2 \land \cdots \land z_{k-1} \xrightarrow{a_k} y$ where each $z_i$ is a fresh variable
    (ii) If $w  = \epsilon$ then $P \defeq (x=y)$.
\AP
By a slight abuse of notation, we write $P\defeq x \xrightarrow{w} y$ to denote such a "$w$-expansion", with $w=a_1\cdots a_k$ where the $z_i$ variables are implicit. For $m \in \Nat$, an ""atom $m$-expansion"" of $x \xrightarrow{L} y$ is a "$w$-expansion" for some $w \in L$ such that $|w| \leq m$. An ""$m$-expansion"" of a "CRPQ" $q$ is the "CQ" resulting from (i) replacing each "atom" with an "atom $m$-expansion" and (ii) removing the "equality atoms" canonically. An ""atom expansion"" is an "atom $m$-expansion" for some $m$, likewise an ""expansion"" is an "$m$-expansion" for some $m$. For instance, for the atom 
$A=(x \xrightarrow{ab^*c} y)$, one example of a "$3$-expansion" would be $x \xrightarrow{a} z_1 \xrightarrow {b}z_2 \xrightarrow{c} y$. $A$ does not have any "$1$-expansion".

\AP
  Let $\intro*\Exp(q)$ denote the (possibly infinite) set of all "expansions" of $q$. and let $\intro* \Expm{m} \subseteq \Exp(q)$ denote the (finite) set of all "$m$-expansions" of $q$.
\AP
  A \reintro{homomorphism} from a ("Boolean") "CRPQ" $q$ to a "graph database" $G=(V,E)$, is defined to be any "homomorphism" $h : \vars(\lambda) \homto G$ for some $\lambda \in \Exp(q)$.
\AP
We further say that $G$ ""satisfies"" $q$, denoted by $G \models q$, if there is such a "homomorphism". We will consider $\Expm{m}$ as a query, which holds true in a "graph database" $G$ if $G \models \lambda$ for some $\lambda \in \Expm{m}$.

\AP
A ""union of CRPQs"" (\reintro{UCRPQ}) is of the form $q \defeq \bigvee_{1\le i \le n} q_i$ , where each $q_i$ is a  "CRPQ". The set of "expansions" of a "UCRPQ" is the union of the "expansions" of the "CRPQs" therein.
Similarly, ""Unions of CQs"" (\reintro{UCQs}) are finite unions of "CQs". 
As mentioned earlier, we assume that "UCRPQs"/"UCQs" are "Boolean", in the sense that they only contain "Boolean" "CRPQs"/"CQs".
A "graph database" "satisfies" a union of "CRPQ" (resp.\ "CQs") if it "satisfies" at least one of its disjuncts.

\AP
For any syntactic object $\+O$, we write $\intro*\vars(\+O)$ to denote the set of variables it contains.
\AP
Let $\intro*\nratoms q$ and $\intro*\nrvars q$ be the number of "atoms" and variables of $q$, respectively.

\AP
Given two "Boolean" queries $q$ and $q'$ (which may be "(U)CQs", "(U)CRPQs", $\Expm{m}$-"expansions", etc.), we write 
$q \intro*\semsubset q'$ if for every "graph database" $G$ such that $G \models q$ we have $G \models q'$, in which case we say that $q$ is ""contained"" in $q'$.
 We say that $q$ and $q'$ are ""equivalent"", written $q \intro*\semequiv q'$, if $q \semsubset q'$ and $q' \semsubset q$.
The following lemma characterizes the "containment" in terms of "expansions" and "homomorphisms".
\begin{lemma}[Folklore]\label{lem:characterization-containment}
  Given two "UCRPQ"s $q$ and $q'$, we have $q \semsubset q'$ if, and only if, for every $\lambda\in \Exp(q)$ there is $\lambda' \in \Exp(q')$ such that there exists a "homomorphism" $\lambda'\homto \lambda$.
\end{lemma}

\AP
A "UCRPQ" $q$ is ""bounded"" if it is "equivalent" to some "UCQ" $\Phi$. The following key property characterizes "boundedness" in terms of "$m$-expansions".
\begin{proposition}\label{characterization_boundedness}
  \cite[Proposition 3]{barcelo2019boundedness}
  A "UCRPQ" $q$ is "bounded" if, and only if, $q$ is "equivalent" to $\Expm{m}$ for some $m \in \Nat$. 
\end{proposition}
\AP
The \intro*\IsBdd ~problem for a class $\+C \subseteq \textup{UCRPQ}$ of queries is the problem of, given a query $q \in \+C$ whether $q$ is "bounded". This is the main problem we will study in this paper. While the "boundedness problem" has been shown to be decidable for "UCRPQ" (as stated below), we will focus on small fragments thereof.

\begin{theorem}\cite[Theorem $11$]{barcelo2019boundedness}
\label{thm-earlier-bdd-result}
\IsBdd ~for "UCRPQs" is "ExpSpace"-complete. Further, the upper-bound holds also in the presence of two-way navigation, and the lower bound holds already for "CRPQs".
\end{theorem}

%% file: results.tex

The lower bound in \Cref{thm-earlier-bdd-result} uses regular expressions that are rather complex. However, as explained in \Cref{sec-Intro}, queries used in practice often contain simple regular expressions. 
Our results focus on "UCRPQs" where atoms are of the form $a^*$, where $a \in \alphabet$ is a label, and more generally of the form $w^*$, where $w \in \alphabet^*$. We next formally define this restricted class of regular expressions.

\AP
Let $\intro*{\reg}(\alphabet)$ denote the set of all regular expressions over $\alphabet$ using concatenation ($\Box\cdot \Box$), union ($\Box + \Box$), Kleene star ($\Box^*$) from the basic letter-expressions ($a$, for each $a \in \alphabet$) and from the symbol $\epsilon$ denoting the empty string. 
\AP
Let $\intro*\aSingleton \subseteq \reg(\alphabet)$ denote regular expressions of the form $a\in \alphabet$ (atomic regular expressions) and $\intro*\aStar$ denote regular expressions of the form $a^*$ where $a\in \alphabet$. 
We write $\aStar[A] \subseteq \reg(\alphabet)$ for some $A\subseteq \alphabet$ to denote the regular expressions of the form $a^*$ where $a\in A$.
Similarly, let $\intro*\wSingleton \subseteq \reg(\alphabet)$ denote regular expressions of the form $w \in \alphabet^*$ ("ie", expressions of the form $w = a_1 \dotsb a_n$ with $a_i \in \alphabet$) and $\intro*\wStar$ denote regular expressions of the form $w^*$ where $w\in \alphabet^*$. 

\AP
Let $\intro*\SF$ (star-free expressions) denote the set of all ""star-free regular expressions"" over $\alphabet$, that is, all expressions formed using $\set{\epsilon, (a)_{a \in \alphabet}, +, \cdot}$. We will further consider the class $\intro*\SSF$ of ""succinct star-free expressions"" which are "star-free regular expressions" which additionally allow for expressions of the form $w^n$ and $w^{\le n}$,  where $w\in \alphabet^*$ and $n$ is encoded in binary. 
The expressions $w^n$ and $w^{\le n}$  are succinct representations of 
$\underbrace{w \dotsb w}_{n \text{ times}}$ and $\underbrace{(\epsilon + w) \dotsb (\epsilon + w)}_{n \text{ times}}$.

\begin{observation}
    \label{prop-basic-inclusion}
    Observe that $\aSingleton \subsetneq \SF \subsetneq \SSF$ and $\aStar \subsetneq \wStar$. However, $\SF$ and $\SSF$ are equivalent in terms of expressive power.
\end{observation}

\AP
The ""size@@SSF"" of an $\SSF$ expression is the number of symbols needed for its encoding, for example the "size@@SSF" of $w^n$ is $|w| + \lceil\log(n)\rceil$. The "size@@SSF" of an expression $w^*$ is $|w|$.

\AP
The ""size"" of an "atom" is the "size@@SSF" of the regular expression therein. The \reintro{size} of a "CRPQ" or "CQ" is the sum of "sizes" of all its "atoms" and the \reintro{size} of a "UCRPQ" or "UCQ" is the sum of the "sizes" of all its "CRPQs" or "CQs".

\medskip
\AP
For any given class $C$ of regular expressions (denoting regular languages), let $\ucrpq{C}$ be the class of "UCRPQs" whose "atoms" contain languages specified by expressions from $C$. We write $\ucrpq{C_1,C_2}$ to denote $\ucrpq{C_1 \cup C_2}$ and $\crpq{C}$ to denote the subclass of $\ucrpq{C}$ consisting of "CRPQs".
For instance, a query in $\crpq{\SSF,\wStar}$ can have an edge label of the form $(ab)^n$ (where $n$ is written in binary) or $(abb)^*$.

The goal of the paper is to prove that  \IsBdd ~for \ucrpq{\SSF,\wStar} is \piPTwo-complete and, further, that the bounded query can be produced efficiently.

\begin{theorem}
\label{thm-is-bdd}
\IsBdd ~for \ucrpq{\SSF,\wStar} is \piPTwo-complete. The problem remains \piPTwo-hard for \crpq{\aSingleton, \aStar}.
    Moreover, if $q\in \ucrpq{\SSF,\wStar}$ is "bounded" then it is equivalent to a query $q'\in \ucrpq{\SSF}$ of linear "size" which can be computed.
\end{theorem}

\AP
For proving \Cref{thm-is-bdd}, we will use the "membership problem" for non-deterministic finite automata (NFA) whose transitions may be succinctly represented as $q \xrightarrow{w^n} q'$, where $n$ is encoded in binary,
which we call "succinct-NFA". We will show in \Cref{sec-Automata} that the "membership problem" for such succinctly represented automata is still in "NP", which is a result of independent interest.

\newcommand{\memsuccinctNFANP}{%
    The "membership problem" for "succinct-NFA" is in "NP".%
}
\begin{theorem}
    \label{prop:mem-succinct-NFA-NP}
    \memsuccinctNFANP
\end{theorem}

\AP
If a query $q$ in \ucrpq{\SSF, \aStar} is not bounded, then it is natural to ask if it can be bounded ``as much as possible''. Intuitively, we want a query $q'$ that is equivalent to $q$ where the "atoms" of $q'$ are either $\SSF$ expressions or $a^*$ only for those letters $a\in \alphabet$ that cannot be bounded in $q$. Formally, for any given $A \subseteq \alphabet$, we say $q$ is $\intro*\bddIn{A}$ if it is equivalent to a query from $\ucrpq{\SSF,\aStar[\bar A]}$  for $\bar A = \alphabet \setminus A$. Observe that such $q$ is "bounded" "iff" it is \bddIn{\alphabet}.
\AP
The problem of checking whether $q \in \ucrpq{\SSF, \aStar}$ is $\bddIn{A}$ is called  the $\intro*\bddInPb{A}$ problem.

\newcommand{\thmMaxbddresult}{%
    For the class \ucrpq{\SSF,\aStar} of queries:
    \begin{enumerate}
    \item \label{thm-Max-bdd-result:1} The $\bddInPb{A}$ problem is \piPTwo-complete.  The problem remains \piPTwo-hard even for \crpq{\aSingleton,\aStar}.
    Moreover, an equivalent query $q' \in \ucrpq{\SSF,\aStar[\bar A]}$ of linear size can be computed.
    \item \label{thm-Max-bdd-result:2}
    There is a unique maximal $A \subseteq \alphabet$ such that $q$ is $\bddIn{A}$, and there is a "PiP2" algorithm for finding it.
    \end{enumerate}
}
\begin{theorem}
    \label{thm-Max-bdd-result}
    \thmMaxbddresult
\end{theorem}

%% file: automata.tex

In this section we study a problem, of independent interest, which will be necessary to obtain our upper bounds (more precisely, the upper bound of \Cref{thm-is-bdd}, covered by \Cref{thm:containment-succinctCQ}).

\AP
Let us define a "succinct-NFA" over an alphabet $\alphabet$ as a classical non-deterministic finite automaton (""NFA"") where transitions can be of the form $q \xrightarrow{w^n} q'$ where $w \in \alphabet^*$ and $n \in \Nat$ is encoded in binary. More formally, a ""succinct-NFA"" over $\alphabet$ is a tuple $(Q,\delta, q_0, F)$, where $q_0 \in Q$ is the ""initial state"", $F \subseteq Q$ is the set of ""final states"", and $\delta \subseteq_{\textit{fin}} Q \times \alphabet^* \times \Nat \times Q$ is a finite set of ""succinct-transitions"". We sometimes write $ q\xrightarrow{w^n}p$ instead of $(q,w,n,p)$, and $n$ is always encoded in binary.
An ""accepting run"" of a given "succinct-NFA" $\automata$ is, as expected, a sequence of transitions $(p_0,w_1,n_1,p_1), (p_1,w_2,n_2,p_2), \dotsc, (p_{m-1},w_m,n_m,p_m)$ such that $p_0$ is  "initial state", and $p_m$ is a "final state". The word from $\alphabet^*$ associated to the "accepting run" is $w_1^{n_1} \dotsb w_m^{n_m} \in \alphabet^*$. The language associated to a given "succinct-NFA" $\automata$, denoted by $\intro*\AutomataLang{\automata}$, is the set of all words associated to "accepting runs".

\AP
Given a word $w =a_1 \dotsb a_n \in \alphabet^*$ over a finite alphabet $\alphabet$, we define the factor of $w$ between $i$ and $j$, denoted by $w\intro*\factor{i}{j}$, as $\epsilon$ if $j \leq i$ or $a_{i+1} \dotsb a_j$ otherwise.

\AP
The ""membership problem"" for "succinct-NFA" is the problem of, given a word $v \in \alphabet^*$, a number $m \in \Nat$ (in binary), and a "succinct-NFA" $\automata$, whether $v^m \in \AutomataLang{\automata}$. 
We observe that this is somewhat close to the automata on compressed strings of \cite[\S 2.3]{MartensNS09}, although in their case, the succinct representation is not on the NFA but only on the input word for testing membership.


\begingroup
\def\thetheorem{\ref{prop:mem-succinct-NFA-NP}}
\begin{theorem}[Restatement]
    \memsuccinctNFANP
\end{theorem}
\addtocounter{theorem}{-1}
\endgroup

\begin{proof}
    Let $\automata$ be a "succinct-NFA",  $v\in \alphabet^*$ be a word and $m$ be a number represented in binary. We assume "wlog" that all numbers in $\delta$ are positive ("ie" there are no $\epsilon$ transitions).

    We first build a new "succinct-NFA" $\automata'$ from $\automata$ so that 
    $v^m \in \AutomataLang{\automata}$ "iff" $\AutomataLang{\automata'}$ contains a word of length $m \times |v|$.
    We define it as follows. 
    The set of states of $\automata'$ is $Q \times \set{0, \dotsc, |v|-1}$.

    There is a "succinct-transition" $(q,i) \xrightarrow{w^n} (p,j)$ if (i) $q \xrightarrow{w^n} p$ is a "succinct-transition" of $\automata$, and (ii) $w^n$ is of the form $v\factor{i}{|v|} \cdot v^\ell \cdot v\factor{0}{j}$ for some $\ell \geq 0$, or equal to $v\factor{i}{j}$.
    The "initial state" of $\automata'$ is $(q_0,0)$ and the set of "final states" is $F \times \set 0$.

    \begin{claim}
        We can build $\automata'$ in polynomial time.
    \end{claim}
    \begin{proof}
        If suffices to prove that checking whether $w^n$ is of the form $v\factor{i}{|v|} \cdot v^\ell \cdot v\factor{0}{j}$ for some $\ell \geq 0$, or equal to $v\factor{i}{j}$ is indeed in polynomial time. 
        Towards this first we build a directed graph $G$ having $\set{0, \dotsc, |v|}$ as vertices, and an edge $i \to j$ if $w = v\factor{i}{|v|} \cdot v^\ell \cdot v\factor{0}{j}$ for some $\ell$. Notice that $\ell$ is bounded by $|w|$ and hence $G$ can be built in polynomial time. 
        
        Now in order to check if $w^n$ is of the form $v\factor{i}{|v|} \cdot v^\ell \cdot v\factor{0}{j}$ for some $\ell \geq 0$, it suffices to check if there is a path of length $n$ from $i$ to $j$ in $G$. 
        The set of all sizes of paths from $i$ to $j$ in $G$ can be seen as the "Parikh-image"\footnote{\AP Remember that the ""Parikh-image"" of a language $L \subseteq \alphabet^*$ over a one-letter alphabet $\alphabet = \set{a}$ is $\set{t \in \Nat : a^t \in L}$.} of an "NFA" over a one-letter alphabet. Since the problem of testing membership of a vector in the "Parikh-image" of an "NFA" is in polynomial time as soon as the alphabet is bounded \cite{KopczynskiT10}, the statement follows.
    \end{proof}

    \begin{claim}
        $v^m \in \AutomataLang{\automata}$ if{f} $\AutomataLang{\automata'}$ contains a word of length $m \times |v|$.
    \end{claim}
    \begin{proof}
        \proofcase{Left-to-right implication.} 
        Assuming $v^m \in \AutomataLang{\automata}$ consider an "accepting run" $(p_0,w_1,n_1,p_1), (p_1,w_2,n_2,p_2),$ $\dotsc,$ $(p_{t-1},w_t,n_t,p_t)$ on $v^m$. Then, for each $i \in \set{1,\dotsc, t}$, we have that $w_1^{n_1} \dotsb w_i^{n_i}$ is a prefix of $v^m$ and thus can be written as $v^{m_i} \cdot v\factor{0}{j_i}$ for some $j_i$ and $m_i$, where in particular $j_t=0$ and $m_t = m$.
        
        Hence, $((p_0,0),w_1,n_1,(p_1,j_1)), ((p_1,j_1),w_2,n_2,(p_2,j_2)),$ $\dotsc,$ $((p_{t-1},j_{t-1}),w_t,n_t,(p_t,j_t))$ is an "accepting run" of $\automata'$ on $v^m$. 

        \proofcase{Right-to-left implication.} 
        First observe that $\AutomataLang{\automata'} \subseteq \AutomataLang{\automata}$ since an "accepting run" of $\automata'$ on a word $w$ induces an "accepting run" of $\automata$ on the same word by projecting the run onto $Q$.

        Observe that an "accepting run"  of $\automata'$ on a word $w$ forces $w$ to be of the form $v^\ell$. 
        Since $\AutomataLang{\automata'}$ has an accepted word of length $m \times |v|$, this forces $\ell=m$, "ie", the accepted word has to be exactly $v^m$. Since $\AutomataLang{\automata'} \subseteq \AutomataLang{\automata}$ by projecting the run onto $Q$, we obtain $v^m \in \AutomataLang{\automata}$.
    \end{proof}

    \begin{claim}
        Testing if $\AutomataLang{\automata'}$ contains a word of a given length is in \np.
    \end{claim}
    \begin{proof}
        By replacing each transition $q \xrightarrow{w^n} p$ in $\automata'$ with $q \xrightarrow{|w| \times n} p$ we obtain a ``succinct one counter automaton'', whose reachability problem is in \np \cite[Theorem~1]{DBLP:conf/concur/HaaseKOW09}.
    \end{proof}
    This concludes the proof of \Cref{prop:mem-succinct-NFA-NP}.
\end{proof}

%% file: ssec-cont-succCQ.tex
\subsection{Containment Problem for Succinct CQs}
A ""succinct CQ"" is a "CRPQ"($\SSF$) whose every "atom" has an expression of the form $w^n$ for $w \in \alphabet^*$, $n \in \Nat$, with the expected semantics. Remember that we assume that $n$ is given in binary and that the size of each "atom" $w^n$ is $|w| + \lceil\log(n)\rceil$, and hence that every "succinct CQ" has an equivalent \AP""corresponding"" "CQ" of at most exponential size ("ie", its only "expansion").
The \AP""containment problem"" for a class $\class$ of queries is the problem of, given two queries $q,q' \in \class$, whether $q \semsubset q'$.

\begin{theorem}\label{thm:containment-succinctCQ}
    The "containment problem" for "succinct CQs" is in "NP".
\end{theorem}
\begin{proof}
    Given two "succinct CQs" $q,q'$ it suffices to check if there is a "homomorphism" from $\tilde q'$ to $\tilde q$, where $\tilde q'$ and $\tilde q$ are the "CQs" "corresponding" to $q'$ and $q$, respectively.
     Towards this, we first non-deterministically ``break'' some "atoms" $x \xrightarrow{w^n} x'$ of $q$ introducing new variables $x \xrightarrow{w_1^{n_1}} y_1 \land  y_1 \xrightarrow{w_2^{n_2}} y_2 \land \dotsb \land y_{k-1} \xrightarrow{w_k^{n_k}} y_k$ where $w^n = w_1^{n_1} \dotsb w_k^{n_k}$ in such a way that we introduce at most $|\vars(q')|$ new variables. To verify that $w^n = w_1^{n_1} \dotsb w_k^{n_k}$, we can build a "succinct-NFA" of the form $p_0\xrightarrow{w_1^{n_1}}p_1\ldots p_{k-1}\xrightarrow{w_k^{n_k}}p_k$ with $p_0$ and $\{p_k\}$ as initial and final states respectively and check if $w^n$ is accepted by this automaton (which can be done in "NP", by \Cref{prop:mem-succinct-NFA-NP}).
    
     This results in a "succinct CQ" $\hat q$ which is "equivalent" to $q$ (in fact, $\tilde q$ is still its "corresponding" "CQ") and of polynomial ---even linear--- size. We now guess a function $f$ from the variables of $q'$ to the variables of $\hat q$. Finally, for each "atom" $x \xrightarrow{w^n} y$ of $q'$, we check if there is a path from $f(x)$ to $f(y)$ in $\hat q$ reading $w^n$. For this, we can see $\hat q$ as a "succinct-NFA" whose "initial state" is $f(x)$ and its set of "final states" is $\set{f(y)}$, and where we check if $w^n$ belongs to its language, which is in "NP" by \Cref{prop:mem-succinct-NFA-NP}.
\end{proof}

%% file: upper.tex

In this section we prove the upper bounds of \Cref{thm-is-bdd}. Formally, we will prove the following statement.

\begin{theorem}\label{thm-pip2-upper-bound}
  The \IsBdd ~problem for \bigQL is in "PiP2". Further, an equivalent $\ucrpq{\SSF}$ query of linear size can be computed.
\end{theorem}

The goal of this section is to prove \Cref{thm-pip2-upper-bound}. From \Cref{prop-basic-inclusion}, the bound applies also to $\ucrpq{\wSingleton,\wStar} \subseteq \ucrpq{\SF,\wStar} \subseteq \ucrpq{\SSF,\wStar}$. 

\AP
For $q \in \bigQL$ and $m \in \Nat$, let $\intro*\qm{m} \in \ucrpq{\SSF}$ be the result of replacing each regular expression of the form $w^*$ in $q$ with $w^{\leq m}$. We begin with the following observation.

\begin{observation}\label{obs:equiv-leq}
  If $q \semequiv \Expm{m}$, then $q \semequiv \qm m$. If $q \semequiv \qm m$ then $q \semequiv \Expm{m \cdot m'}$ for $m'$ the maximum length of a word $w$ in an expression of the form $w^*$. Hence, by \Cref{characterization_boundedness}, $q \in \bigQL$ is "bounded" if, and only if, $q$ is "equivalent" to $\qm{\ell}$ for some $\ell \in \Nat$. 
\end{observation}
In view of the previous observation, we will rather work with the more intuitive query $\qm{m}$ rather than $\Expm{m}$.
For economy of space we will simply write ``$\lambda \in \qm{m}$'' to denote $\lambda \in \Exp(\qm{m})$.

\noindent{\bf {Recursive and non-recursive atoms}}. 
An "atom" of the form $x \xrightarrow{w^*} y$  is called ""recursive"". 
Let $\intro*\Recwords \subseteq \alphabet^*$ be the set of all words $w$ from 
the labels $w^*$ of "recursive" "atoms" of $q$. Let $\intro*\Nbound \in \Nat$ be the maximum length of a word in a "non-recursive" atom of $q$.
We will prove \Cref{thm-pip2-upper-bound} via the following intermediate results.

\begin{lemma}\label{lem:new:q_equiv_q(n)}
  If $q \in \bigQL$ is "bounded", then it is "equivalent" to $\qm{\Zbound}$ for $\intro*\Zbound \defeq \nratoms[3]{q} \cdot \Nbound \cdot \nrvars q \cdot \Pi_{w \in \Recwords} |w|$.
\end{lemma}

\begin{lemma}\label{lem:new:generating_algo_pi2p}
  A query $q\in \bigQL$ is "bounded" "iff" $\qm{\Zbound}$ is "equivalent" to $\qm{\Zboundplus}$, for $\intro*\Zboundplus \defeq \nratoms q \cdot \Zbound + 1$.
\end{lemma}

\Cref{lem:new:generating_algo_pi2p} reduces the "boundedness problem" for \bigQL\ to the "containment problem" for "succinct CQs". The proof of \Cref{lem:new:generating_algo_pi2p} will use \Cref{lem:new:q_equiv_q(n)}. Assuming these two results, we will first prove \Cref{thm-pip2-upper-bound}.

\begin{proof}[Proof of \Cref{thm-pip2-upper-bound}]
  From \Cref{lem:new:generating_algo_pi2p}, to test whether a query $q$ is "bounded" 
  can be reduced to checking 
  whether $\qm{n} \semequiv \qm{n'}$ for some $n,n'\in \Nat$ of polynomial space (in binary). Given that the ``succinct'' labels of $\SSF$ expressions appearing in $q$ have the form $w^{\leq m}$ or $w^{m}$ with $m$ in binary, the "expansions" of $\qm n$ and  $\qm{n'}$   are at most single-exponential in the "size" of $q$. Moreover, every "expansion" of $\qm n$ can be expressed as a "succinct CQ" of polynomial "size". 
  
  In order to check that $\qm n \semsubset \qm {n'}$ does not hold, we first guess an "$n$-expansion" $\lambda$ of $q$ that is not contained in $\qm{n'}$, as a  "succinct CQ" of polynomial "size". This is possible by the discussion above. Then we check that $\lambda \semsubset \qm{n'}$ does not hold.
To check $\lambda \semsubset \qm{n'}$,  it suffices to guess an "$n'$-expansion" $\lambda'$ of $q$, which, once again,  we assume is a "succinct CQ" of polynomial "size", and then check that $\lambda \semsubset \lambda'$. This is in "NP" by \Cref{thm:containment-succinctCQ}; thus, we have a procedure which is in  "coNP" for testing $\lambda \not\semsubset \lambda'$.

 This gives a "SigmaP2" algorithm to test  $\qm{n} \not\semsubset \qm{n'}$, in other words a "PiP2" algorithm to test $\qm{n} \semsubset \qm{n'}$, and thus whether $q$ is "bounded".

Finally, the existence of the equivalent query is trivial since $\qm{\Zbound}$ can be produced in linear time.
  \end{proof}

Hence it is enough to prove \Cref{lem:new:generating_algo_pi2p} which in turn needs \Cref{lem:new:q_equiv_q(n)}. The next two subsections are focused on the proofs of these two lemmas.

\input{upperbound-firstlemma}
\input{upperbound-secondlemma}

%% file: upperbound-firstlemma.tex
\subsection{Proof of \Cref{lem:new:q_equiv_q(n)}}

We recall the statement of \Cref{lem:new:q_equiv_q(n)}:
\begingroup
  \def\thetheorem{\ref{lem:new:q_equiv_q(n)}}
  \begin{lemma}
    If $q \in \bigQL$ is "bounded", then it is equivalent to $\qm{\Zbound}$.
  \end{lemma}
  \addtocounter{theorem}{-1}
\endgroup

Before proving the result formally, we describe the key ideas informally. Suppose $q$ is "bounded", then by \Cref{obs:equiv-leq}, it is "equivalent" to $\qm{M}$ for some $M$. Assume for contradiction that this $M$ is necessarily larger than $\Zbound$. This means that there exists some $\lambda \in \Exp(q)$ such that for every $\hat \lambda \in \qm{\Zbound}$ there is no "homomorphism" from $\hat \lambda$ to $\lambda$ but there is some  $\lambda' \in \qm{M}$ (of minimal size) such that there is a "homomorphism" $h: \lambda'\homto \lambda$. 
Since $\lambda'\notin \qm{\Zbound}$ there is some "atom" $x\xrightarrow{w^*}y$ in $q$ such that it is expanded as $x\xrightarrow{w^l}y$ in $\lambda'$ for some $l > \Zbound$. The choice of $\Zbound$ is large enough such that we can argue that using the "homomorphism" $h:\lambda'\homto \lambda$, we can transform $\lambda'$ into some $\lambda''$ (by contracting some "recursive" "expansions") such that there is a "homomorphism" from $\lambda''$ to $\lambda$. This would contradict the minimality of $\lambda'$. The main technical difficulty is then to produce $\lambda''$ and the "homomorphism" from $\lambda''$ to $\lambda$.

To prove \Cref{lem:new:q_equiv_q(n)} we shall need the following claim.

\begin{claim}\label{cl:boundimage}
        There is $\lambda' \in \qm{M}$ and $h : \lambda' \homto \lambda$ such that the size of the image
        $\Image(h)$ of $h$ is at most $\Zbound$.
    \end{claim}

  We first prove that \Cref{cl:boundimage} implies \Cref{lem:new:q_equiv_q(n)}.

\begin{proof}[Proof of \Cref{lem:new:q_equiv_q(n)}]    
    By means of contradiction, suppose $q$ is "bounded" but not equivalent to $\qm{\Zbound}$. Then, there is no $\hat \lambda \in \qm{\Zbound}$ such that there is a homomorphism $\hat \lambda \homto \lambda$.
    By \Cref{cl:boundimage}, there is some $\lambda' \in \qm{M}$ and $h : \lambda' \homto \lambda$ such that the size of the image $\Image(h)$ of $h$ is at most $\Zbound$. Pick a minimal\footnote{By minimal $\lambda'$ we mean in terms of the number of variables.} $\lambda'$ that satisfies this property. 
    
    By assumption, $\lambda'\notin \qm{\Zbound}$. Hence, there is some "recursive" "atom expansion" $x \xrightarrow{w^\ell} y = x \xrightarrow{w} x_1 \xrightarrow{w} \dotsb \xrightarrow{w} x_{\ell-1} \xrightarrow{w} y$ of $\lambda'$ where $\ell > \Zbound$.\footnote{Remember that a "recursive" "atom expansion" of $\lambda'$ is simply an "expansion" of a "recursive atom" of $q$ in $\lambda'$.} By \Cref{cl:boundimage} and the pigeonhole principle, there will be two variables $x_i,x_j$ such that $h(x_i) = h(x_j)$ for some $i<j$. But then we can replace the "atom expansion" with $x \xrightarrow{w^{\ell-(j-i)}} y = x \xrightarrow{w} x_1 \xrightarrow{w} \dotsb \xrightarrow{w} x_{i} \xrightarrow{w} x_{j+1} \xrightarrow{w} \dotsb \xrightarrow{w} x_{\ell-1} \xrightarrow{w} y$, obtaining an expansion $\lambda'' \in \qm{M}$ which is strictly smaller than $\lambda'$ such that $h$ restricted to $\lambda''$ forms  a "homomorphism" $h':\lambda'' \homto \lambda$. Moreover, since $h' \subseteq h$, the size of $\Image(h')$ is still at most $\Zbound$, contradicting the minimality of $\lambda'$.
\end{proof}

\begin{figure}[t]
\centerline{
	\includegraphics[width=.49\textwidth]{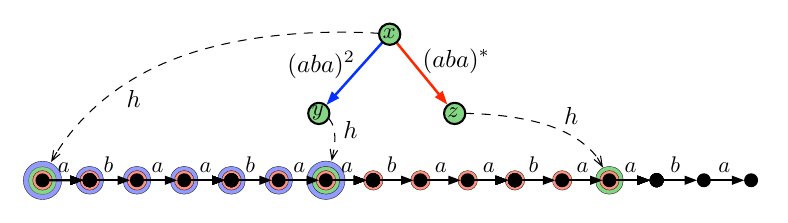}}
	\caption{Consider the query $x \xrightarrow{(aba)^n} y \wedge x \xrightarrow{(aba)^*} z$, where $(aba)^n$ is a $\SSF$ with $n=11$ in binary ("ie", 3).. 
    Below, we have the corresponding 
	coloring scheme in $\lambda^+$ with the respective colors according to a "homomorphism" $h$.}
	\label{fig:color}
\end{figure}

Thus it remains to prove \Cref{cl:boundimage}.
    \begin{proof}[Proof of \Cref{cl:boundimage}]
        Let $M'$ be the size of $\qm{M}$ (seen as a "UCQ").
        Let $\lambda^+$ be the result of replacing each atom expansion $x \xrightarrow{w^\ell} y$ of $\lambda$ such that $\ell > \Zbound$ with $x \xrightarrow{w^{M'}} y$.
        Notice that $\lambda^+ \in \Exp(q)$ ---indeed, "atom expansions" of this kind can only come from   "recursive atoms" (because $\Zbound$ is sufficiently large). Now pick a $\lambda'\in \qm{M}$ such that $\lambda'$ is of minimal "size" and there is a "homomorphism" $h : \lambda' \homto \lambda^+$.
        We show the following:

        \begin{enumerate}[(1)]
            \itemAP $|\Image(h)| \leq \Zbound$; \label{eq:image-bound}
            \itemAP $h$ is also a "homomorphism" to (an "isomorphic copy"\footnote{By \AP""isomorphic copy"" of $\lambda$ we mean an "expansion" $\tilde\lambda$ such that there is a bijection $f$ between the variables 
            of $\lambda$ and $\tilde\lambda$ such that $x  \xrightarrow{a} y$ is an "atom" of $\lambda$ "iff" 
            $f(x)  \xrightarrow{a} f(y)$ is an "atom" of $\tilde\lambda$.} of) $\lambda$, "ie",  $h : \lambda' \homto \lambda$.
            \label{eq:h-witness-hom}
        \end{enumerate}

        \begin{figure*}
            \centering
       \includegraphics[width=\textwidth]{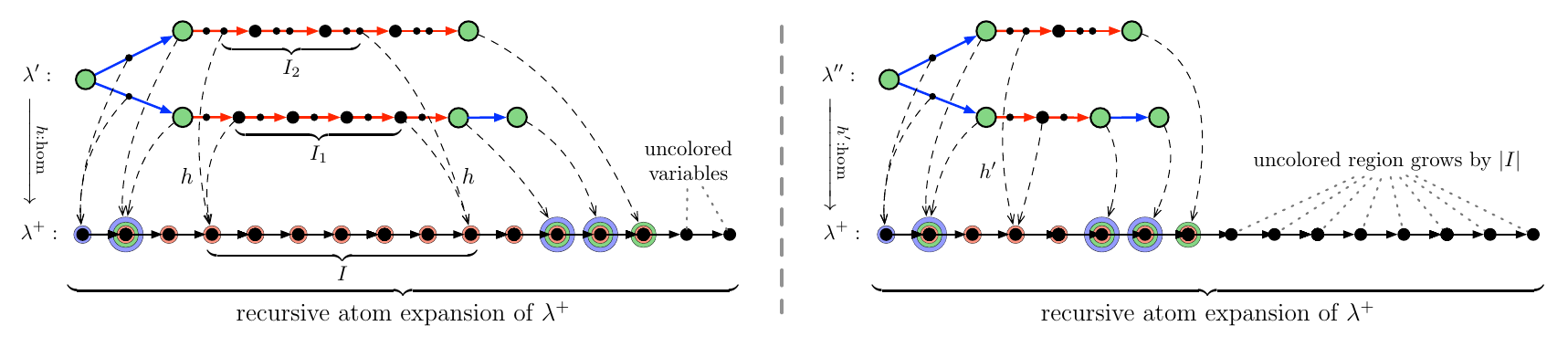}
            \caption{(Left) An example of a "homomorphism" from $\lambda'$ to $\lambda^+$ with a large "purely-red" interval $I$. 
            In $\lambda'$, red (resp.\ blue) edges correspond to "expansions" of "recursive atoms" (resp.\ "non-recursive" "atoms"), and green vertices correspond to variables of $q$.
            The $h$-preimage of $I$ contains two intervals $I_1$ and $I_2$ coming from recursive atom expansions of words of length 2 and 3 respectively.             \\
            (Right) The resulting $\lambda''$. The intervals $I_1, I_2$ are contracted to a single variable in the "homomorphism" to $\lambda^+$. The image of the variables which appeared to the right of $I$ are now shifted to the left, making the number of uncolored variables grow by $|I|-1$.}
            \label{fig:shrinkhom}
        \end{figure*}
        
        \proofcase{\eqref{eq:image-bound}}
        Let us start with the "expansion" $\lambda' \in \qm{M}$ such that $h : \lambda' \homto \lambda^+$. First, we devise a color scheme for the variables of $\lambda^+$.

        Let us color with \textit{blue} any variable of $\lambda^+$ which is the $h$-image of a variable of a "non-recursive" "atom expansion" of $\lambda'$. Color a variable \textit{green} if it is the $h$-image of a variable of $q$, and color it \textit{red} if it is the $h$-image of a variable from a "recursive"  "atom expansion". Of course, some variables may be colored with multiple colors (see Figure \ref{fig:color}).

        Consider any "recursive" "atom expansion" of $\lambda^+$ of the form $x \xrightarrow{w^{M'}} y$. Recall that $M'$ is the size of the $M$-expansion $q(M)$
         considered as a UCQ, and $\lambda' \in q(M)$. Since $M'$ is larger than any "expansion" of $\qm{M}$ (in particular $\lambda'$), the "expansion" $x \xrightarrow{w^{M'}} y$ must contain some uncolored variables. Let $x_0, \dotsc, x_{M'}$ be the variables of the path $x \xrightarrow{w^{M'}} y$, in order, with $x_0 = x$ and $x_{M'}=y$. 
        \AP
        If there is an ""interval"" $I= \set{x_i,x_{i+1}, \dotsc, x_j}$ with $j-i = \intro*\ZredBound \defeq \Pi_{w \in \Recwords} |w|$ which is ""purely-red"" ("ie", whose every variable is colored red and by no other color), this means that the $h$-preimage of $I$ consists of "intervals" of "recursive" "atom expansions" of the form $\tilde w^k$ for some $\tilde w$ and $k$ such that $|\tilde w^k| = \ZredBound$ (note that $\tilde w$ need not be a "recursive atom", it could be some word in $\alphabet^*$). 
        Concretely, $h^{-1}(I) = I_1 \cup \dotsb \cup I_r$ such that each $I_s$ is of the form $I_s = \set{z_{(s,0)}, \dotsc, z_{(s,\ZredBound)}}$ where $h(z_{(s,t)}) = x_{i+t}$ for every $s,t$ and $I_s$ forms a directed path reading some $\tilde w^k$.

        \AP
        Consider $\lambda''$ as the result of ""contracting"" these "intervals" $I_s$ from $\lambda'$: for each interval $I_s$, we remove all atoms of $\lambda'$ inside $I_s$, and we rename $z_{(s,0)}$ as $z_{(s,\ZredBound)}$(see Figure \ref{fig:shrinkhom}). 
        Since this operation corresponds to removing some iterations 
        of a "recursive" "atom expansion" (and \emph{only} "recursive" "atom expansion" since the color is \emph{exclusively} red), we have $\lambda'' \in \qm{M}$.

        Further, note that since there is always at least one uncolored variable $\lambda^+$ (which is a "recursive" "atom expansion"),
        we have
        $\lambda'' \homto \lambda^+$ because it suffices to shift the image of the variables appearing of one side of $I$ as shown in \Cref{fig:shrinkhom}, in particular making the uncolored region grow.
        This is in contradiction with $\lambda'$ being minimal.

        Hence, in an "atom expansion" $x \xrightarrow{w^{M'}} y$ of $\lambda^+$, there cannot be a "purely-red" "interval" of size $\ZredBound$.      On the other hand, since $\nrvars q$ is the number of variables in $q$, there cannot be more than $\nrvars q$ green-colored variables. Finally, there cannot be more than $\nratoms q \cdot \Nbound$ blue-colored variables.
        
        \AP
        Thus, there cannot be a colored "interval" of size greater than $\intro*\ZcolBound \defeq \nratoms q \cdot \Nbound \cdot \nrvars q \cdot \ZredBound$ (as it would imply a $\ZredBound$-sized "purely-red" "interval"). Since $\nratoms q$ is the number of atoms, 
         there cannot be more than $\nratoms q$ colored intervals of maximal size in such "atom expansions" $x \xrightarrow{w^{M'}} y$  (actually, the number of connected components of $q$ would suffice).

        \AP
        How many colored variables does $\lambda^+$ have? Since the maximum length of a colored interval is $\ZcolBound$ and since there are no more than $\nratoms q$ "intervals" on any given "atom expansion" of $\lambda^+$, there cannot be more than 
        $\Zbound = \nratoms[2]{q} \cdot \ZcolBound$ colored variables in $\lambda^+$. Hence, $|\Image(h)| \leq \Zbound$.

        \proofcase{\eqref{eq:h-witness-hom}}
        Next, we show that $h$ is also a "homomorphism" to an "isomorphic copy" of $\lambda$. Observe that if we have an uncolored "interval" of $x \xrightarrow{w^{M'}} y$ of size $|w|$ we can simply "contract" it obtaining some $\lambda^\pm$ such that $h$ is still a "homomorphism" $\lambda' \homto \lambda^\pm$ and $\lambda^\pm \in \Exp(q)$.

        \AP
        Consider any "recursive" "atom expansion" from $\lambda^+$ of the form $x \xrightarrow{w^{M'}} y$. By construction of $\lambda^+$, recall that  $x \xrightarrow{w^{M'}} y$ was obtained by replacing some "atom expansion" 
        $x \xrightarrow{w^{m}} y$ of $\lambda$ where $m>\Zbound$.
        If we could find $M' - m$ such uncolored "intervals" of an "atom expansion" $x \xrightarrow{w^{M'}} y$ as before, we could "contract" them and obtain (an "isomorphic copy" of) $x \xrightarrow{w^m} y$. Further, $h$ would still be a valid "homomorphism".

        But how many such uncolored $|w|$-sized "intervals" are there?
        The worst-case scenario in which we could not find any such uncolored "interval" is the one where every pair of colored-intervals (and there are at most $\nratoms q$ many) separated by $|w|-1$ uncolored-intervals. Since the size and number of colored-"intervals" is bounded by $\ZcolBound$ and $\nratoms q$ respectively, the length of such worst-case "atom expansion" is bounded by
        \knowledgenewrobustcmd{\Zworst}{\cmdkl{Z_{\textit{worst}}}}
        \AP
        $\intro*\Zworst \defeq \nratoms q  \cdot \ZcolBound + (\nratoms q +1) \cdot (|w|-1)$.
        Hence, there are at least $\lceil ((M' \cdot |w|) - \Zworst)/|w| \rceil$ uncolored-"intervals" of size $|w|$, and since
        \begin{align*}
            &\lceil ((M' \cdot |w|) - \Zworst)/|w| \rceil \\
            \geq{}& M' - \Zworst/|w|\\   
            \geq{}& M' - \nratoms q  \cdot \ZcolBound \cdot |w|\\
            \geq{}& M' - \Zbound   \tag{since $\Zbound > \nratoms q  \cdot \ZcolBound \cdot |w|$}\\
            \geq{}& M' - m,\tag{since $m>\Zbound$}
        \end{align*}
        we can "contract" $(M'-m)$-many uncolored intervals among the available ones.
        If we repeat these contractions for every $M'$-"atom expansion", we obtain (an "isomorphic copy" of) $\lambda$ from $\lambda^+$, which shows $h : \lambda' \homto \lambda$.
    \end{proof}
    This concludes the proof of \Cref{cl:boundimage} and hence of \Cref{lem:new:q_equiv_q(n)}.

%% file: upperbound-secondlemma.tex
\subsection{Proof of \Cref{lem:new:generating_algo_pi2p}}
We will actually prove the following statement, which entails \Cref{lem:new:generating_algo_pi2p}.
\begin{lemma}\label{lem:bound-both-sides}
    For every query $q \in \bigQL$, the following are equivalent:
    \begin{enumerate}[(1)]
        \item \label{lem:bound-both-sides:1} $q$ is "bounded",
        \item \label{lem:bound-both-sides:2} $q \semequiv \qm{\Zbound}$,
        \item \label{lem:bound-both-sides:3} $\qm{\Zboundplus} \semequiv \qm{\Zbound}$.
    \end{enumerate}
\end{lemma}
\begin{proof}
    \proofcase{\eqref{lem:bound-both-sides:2} $\Rightarrow$ \eqref{lem:bound-both-sides:1}} is by definition of being "bounded".

    \proofcase{\eqref{lem:bound-both-sides:1} $\Rightarrow$ \eqref{lem:bound-both-sides:3}} By the previous \Cref{lem:new:q_equiv_q(n)}, if $q$ is "bounded" then it is "equivalent" to $\qm{\Zbound}$, and hence also to $\qm{\Zboundplus}$ since $\Zboundplus \geq \Zbound$.

    \proofcase{\eqref{lem:bound-both-sides:3} $\Rightarrow$ \eqref{lem:bound-both-sides:2}}
    We show the contrapositive statement, namely that $q \not\semsubset \qm{\Zbound}$ implies $\qm{\Zboundplus} \not\semsubset \qm{\Zbound}$.
    Assume that $q \not\semsubset \qm{\Zbound}$. Then  there is some "expansion" $\lambda \in \Exp(q)$ such that 
    $\lambda \not\semsubset\qm{\Zbound}$.
     Observe that the size of every "expansion" from $\qm{\Zbound}$ is strictly bounded by 
    $\Zboundplus = \nratoms q \cdot \Zbound + 1$. Consider the "expansion" $\lambda'$ as the result of replacing each "atom expansion" $x \xrightarrow{w^\ell} y$ of $\lambda$ such that $\ell > \Zbound$ with $x \xrightarrow{w^{\Zboundplus}} y$. Clearly, $\lambda'\in \qm{\Zboundplus}$.    

    We want to show that 
    $\lambda' \not\semsubset \qm{\Zbound}$. 
    By means of contradiction, assume that there is an "expansion"  $\hat\lambda \in \qm{\Zbound}$ such that $h:\hat\lambda \homto \lambda'$. 
    Then, each replaced "atom expansion" $x \xrightarrow{w^{\Zboundplus}} y$ of $\lambda'$ must contain at least one vertex which is not in the image of $h$.
    Then, we can expand it back\footnote{Remember that $\ell > \Zboundplus$ and we can add sufficiently many new vertices in the neighborhood of the vertex that is not in the image of $h$ to obtain $x \xrightarrow{w^{\ell}} y$.} to $x \xrightarrow{w^{\ell}} y$ 
    still results in a "homomorphism". Therefore, $\hat\lambda \homto \lambda'$ implies $\hat\lambda \homto \lambda$; and since 
    $\lambda \not\semsubset\qm{\Zbound}$, we have $\lambda'\not\semsubset\qm{\Zbound}$.
    
    Hence, $q \not\semsubset \qm{\Zbound}$ implies $\qm{\Zboundplus} \not\semsubset \qm{\Zbound}$.
\end{proof}

%% file: lower-arxiv.tex
Here we show that even for the simplest fragment, namely "CRPQ" where the regular expressions are of the form $a^*$ or just $a$, the "boundedness problem" is already "PiP2"-hard. 

\begin{theorem}\label{thm:lowerbound} 
The \IsBdd ~problem for "CRPQ"$(\aSingleton,\aStar)$ is "PiP2"-hard.
\end{theorem}

    \begin{observation}\label{rk:differences-kr20}
        The queries $q_1$ and $q_2$ defined in the proof of \Cref{thm:lowerbound} are constructed in the same way as the queries $Q_1$ and $Q_2$ of  \cite[proof of Theorem 4.3]{FigueiraGKMNT20}, except that the following changes have been made:
    \begin{enumerate}
        \item $\bullet \xrightarrow{t}\bullet$ has been replaced with $\bullet\xleftarrow{b}\bullet\xleftarrow{a}\bullet$;
        \item $\bullet \xrightarrow{f}\bullet$ has been replaced with $\bullet\xleftarrow{b}\bullet\xrightarrow{a}\bullet$;
        \item $\bullet\xrightarrow{t+f}\bullet$ in $D$ has been replaced with $\bullet\xleftarrow{b}\bullet\xleftarrow{a^*}\bullet\xrightarrow{a}\bullet$;
        \item a self-loop on $s$ has been added to all roots of the $D/E$-trees in $q_1$;
        \item a self-loop on $s\cdot s^*$ (same as $s^+$) has been added to all roots of each clause sub-query in $q_2$. 
    \end{enumerate}
    \end{observation}

The proof goes by reduction from $\forall\exists$-QBF satisfiability. Consider an instance of $\forall\exists$-QBF 
\begin{align}
    \Phi=\forall x_1,\dots,x_n\exists y_1,\dots,y_l\:\varphi(x_1,\dots,x_n,y_1,\dots,y_l)    \label{eq:qbf-form}
\end{align}
where $\varphi$ is quantifier free and in 3-CNF.
We define "Boolean" "CRPQ"$(\aSingleton,\aStar)$ queries $q_1$ and $q_2$ on the alphabet $\alphabet=\set{b,a,s,j,x_1,\dots,x_n,y_1,\dots,y_l}$ as depicted in Figures~\ref{fig:q_1} and \ref{fig:q_2} respectively. Note that $q_1$ does not depend on the structure of $\varphi$ and $q_2$ encodes every clause occurring in $\varphi$ as a disjoint gadget. We will prove that $\Phi$ is satisfiable "iff" the "CRPQ" $q_1\land q_2$ is bounded.
\begin{figure}
    \includegraphics[scale=.7]{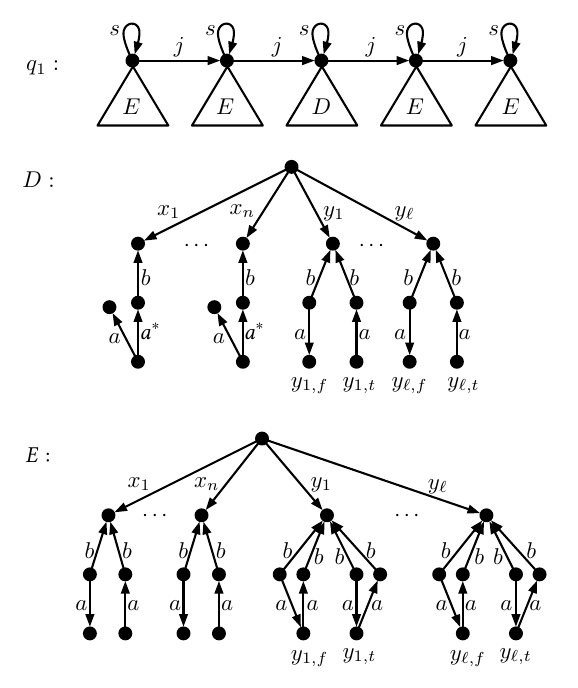}
    \centering
    \caption{Query $q_1$ used in the proof of \Cref{thm:lowerbound}. Variables with identical $y_{i,\alpha}$ label of the gadgets $D$ and $E$ (across all $E$) represent the \emph{same} variable ("eg", $y_{1,f}$ in $D$ and $E$ are the same variable).}
        \label{fig:q_1}
  \end{figure}

The construction is similar to the proof of  \cite[Theorem 4.3]{FigueiraGKMNT20}. However, the context is different since the cited theorem proves "PiP2"-hardness for the "containment problem" for the fragment of "CRPQs" where disjunctions of letters ("ie", expressions like ``$a+b$'') are allowed.

\begin{claim}\label{b_a_a*_bounded}
    The "Boolean" query $q = y\xleftarrow{a}x\xrightarrow{a^*}z \xrightarrow{b} w$
    is "bounded". More precisely, $q\semequiv  \Expm{1}$, "ie" $q$ is equivalent to the disjunction of patterns $\bullet\xleftarrow{b}\bullet\xrightarrow{a}\bullet$ and $\bullet\xleftarrow{b}\bullet\xleftarrow{a}\bullet$.
\end{claim}
\begin{proof}
Let $\lambda_m$ denote the "$m$-expansion" of $a^*$ in $q$.
Note that $\lambda_0$ and $\lambda_1$ are equivalent to the two patterns $\bullet\xleftarrow{b}\bullet\xrightarrow{a}\bullet$ and $\bullet\xleftarrow{b}\bullet\xleftarrow{a}\bullet$, respectively.
Also observe that $\lambda_1 \homto \lambda_n$ for every $n>0$ via the "homomorphism" $\set{w \mapsto w, z \mapsto z, y \mapsto z, x\mapsto z'}$, where $z'$ is the variable such that $z' \xrightarrow{a} z$ is an atom of $\lambda_n$.
Hence, $q$ is equivalent to the disjunction of $\lambda_0$ and $\lambda_1$.
\end{proof}

Note that this query $q$ in \Cref{b_a_a*_bounded} is attached as a `tail' for each of the $x_i$ variable in the $D$ gadget of $q_1$. 
As a consequence we have the following.
\begin{claim}\label{q1_bounded}
    $q_1$ is "bounded".
\end{claim}
\begin{proof}
    From \Cref{b_a_a*_bounded}, we have that 
    $\bullet\xleftarrow{b}\bullet\xleftarrow{a}\bullet\xrightarrow{a}\bullet$ can be embedded into 
    any non-zero "expansion" of $\bullet\xleftarrow{b}\bullet\xleftarrow{a^*}\bullet\xrightarrow{a}\bullet$. 
    Since these gadgets are ``leafs'' in the definition of $q_1$, it follows that for any "expansion" $\lambda$ of $q_1$ there is a corresponding "expansion" of $\Expm 1$ (choosing one of the two gadgets depending on whether the corresponding "expansion" was $0$ or $>0$) which can be "homomorphically" mapped to $\lambda$. Thus, the claim follows.
\end{proof}

 \Cref{thm:lowerbound} is a consequence of the following lemma.

\begin{lemma}
\label{lem-tfae-sat-cont-bdd}
Let $\Phi$ be an instance of $\forall\exists$-QBF and $q_1$ and $q_2$ as described. Then the following are equivalent:

(1) $\Phi$ is satisfiable,
(2) $q_1\semsubset q_2$, and
(3) $q_1\land q_2$ is "bounded".
\end{lemma}
\begin{proof}
    \proofcase{$(1) \Leftrightarrow (2)$}(Proof sketch.)\footnote{For a formal proof along similar lines, we refer the reader to \cite[Theorem 3]{DBLP:journals/corr/abs-2003-04411}} This follows by an adaptation of the proof of Theorem 4.3 of \cite{FigueiraGKMNT20}.%
    \footnote{Indeed, it suffices to reproduce the cited proof by replacing each occurrence of $\bullet\xrightarrow{t}\bullet$ with $\bullet \xleftarrow{b}\bullet\xleftarrow{a}\bullet$, and each occurrence of $\bullet\xrightarrow{f}\bullet$ with $\bullet \xleftarrow{b}\bullet\xrightarrow{a}\bullet$.}

    The proof goes by  reduction from $\forall\exists$-QBF satisfiability. Given a $\forall\exists$-QBF formula $\Phi$ as in \eqref{eq:qbf-form}, we construct queries $q_1$ and $q_2$ as depicted in \Cref{fig:q_2,fig:q_1}. 
    The query $q_2$ consists of gadgets as in \Cref{fig:q_2} per clause, while $q_1$ is defined in \Cref{fig:q_1}.

    The ``tail parts'' of the edges labeled $x_i$  in the $D$ gadget allows any assignment to the variable $x_i$, that is,  the $a^*$ in $D$  can be expanded for a `true' assignment $\bullet \xleftarrow{b}\bullet\xleftarrow{a}\bullet$ or shrunk for a `false' assignment $\bullet \xleftarrow{b}\bullet\xrightarrow{a}\bullet$. The valuation for the $y_i$ variables comes from $q_2$ (as in \Cref{fig:q_2}):  the $y_{i,tf}$ node  in $q_2$ embeds uniquely into one of the $y_{i,t}$ or $y_{i,f}$ nodes of $q_1$ to witness the containment.  
    
    When the formula is satisfiable, there exists a valuation for all the $y_i$ variables for any valuation of the $x_i$ variables. The $D$ gadgets allows the choice of any true/false assignments for $x_i$ variables, and we can embed the  $y_{-,tf}$ nodes of $q_2$ into exactly one of the $y_{-,t}$, $y_{-,f}$ nodes of $q_1$ (depending on whether the $y_i$ takes on true or false to make the formula true). Thus, this gives the containment of $q_1$ in $q_2$. 
    
    For the converse, assume the containment. Then we have an embedding of each clause of  $q_2$ into the given graph database for a given choice of ``tail parts'' of the $x_i$ labeled edges of $D$. In particular, we can always map a literal in each clause of $q_2$ to $D$: this ensures satisfiability of the formula.

    \smallskip

    \proofcase{$(2) \Leftrightarrow (3)$}
    For the left-to-right direction, suppose $q_1\semsubset q_2$. Note that we always have $q_1\land q_2 \semsubset q_1$ and since $q_1 \semsubset q_2$ we also have $q_1 \semsubset q_1\land q_2$. Thus, $q_1\land q_2 \semequiv q_1$. From \Cref{q1_bounded}, $q_1$ is "bounded" and hence $q_1\land q_2$ is also "bounded".

    For the right-to-left direction, assume $Q=q_1\wedge q_2$ is "bounded". In view of \Cref{characterization_boundedness}, it is equivalent to $\Expm[Q]{m}$, for some $m$.
    We show that this implies $q_1\semsubset q_2$.
    By contradiction, assume there is an "expansion" $\lambda_1$ of $q_1$ such that for all "expansions" $\lambda_2$ of $q_2$, $\lambda_2$ does not "homomorphically" map to $\lambda_1$.

    Now consider the "expansion" $\lambda_2$ of $q_2$ in which all $s^*$ "atoms" are expanded exactly $m+1$ times ---"ie",
    $\bullet \xrightarrow{s^*} \bullet$ is replaced with an $s^{m+1}$-path and hence we get a loop on $s^{m+2}$ at each `base' node in $q_2$.

    Let $\lambda=\lambda_1\wedge \lambda_2$, which is an "expansion" of $Q$. 
    Since $Q$ is equivalent to $\Expm[Q]{m}$, we have $\lambda\semsubset \Expm[Q]{m}$. In other words, there is an "expansion" $\hat\lambda$ of $\Expm[Q]{m}$  and a "homomorphism" $h: \hat\lambda \homto \lambda$. Note that $\hat\lambda$ must be of the form $\hat\lambda= \hat\lambda_1\wedge \hat\lambda_2$, where $\hat\lambda_1$ is an "expansion" of $q_1$ and $\hat\lambda_2$ is an "expansion" of $q_2$ where every $s$-cycle in $\hat\lambda_2$ is of length at most $m+1$. 
    
    But then $\hat\lambda_2\homto \lambda_1$ is not possible by the choice of $\lambda_1$. This implies that we have $\hat\lambda_2 \homto\lambda_2$ which means that an $s$-cycle of length $\leq m+1$ is "homomorphically" mapped to an $s$-cycle of length $m+2$, which is not possible. Hence, no such "homomorphism" $h$ can exist, which is a contradiction.
\end{proof}

    \begin{figure}
        \includegraphics[scale=.7]{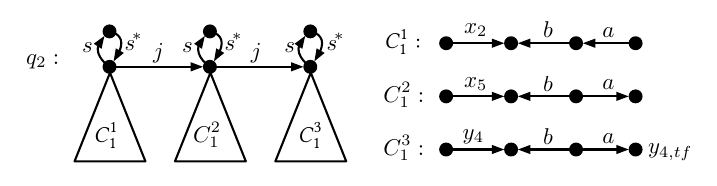}
        \centering
        \caption{An example for query $q_2$, used in the proof of \Cref{thm:lowerbound}, for the clause $(x_2\vee \neg x_5\vee\neg y_4)$. There will be one such gadget for every clause of the formula in $q_2$. Variables having identical $y_{i,\textit{tf}}$-label represent the \emph{same} variable. Only the final variable of paths representing $y_i$-variables from $\Phi$ may have a $y_{i,\textit{tf}}$-label.}
        \label{fig:q_2}
    \end{figure}

%% file: sec-A-boundedness-arxiv.tex
\section{Boundedness by letter}
\label{sec:letter-boundedness}

If a query $q$ in \ucrpq{\SSF, \aStar} is not "bounded", we can still try to bound it in as many "atoms" as possible, this is the object of the $\bddInPb{A}$ problem. Recall that $q \in \ucrpq{\aSingleton,\aStar}$ is $\bddIn{A}$ if it is equivalent to a $\ucrpq{\aSingleton,\aStar[\bar A]}$ query, with $\bar A = \alphabet \setminus A$. 

\AP
For $A \subseteq \alphabet$ and $m \in \Nat$, let $\intro*\qAm{A,m}$ be the result of replacing each $a^*$ s.t.\ $a \in A$ with $a^{\leq m}$. 
We still have the characterization along the lines of \Cref{characterization_boundedness}.

\begin{proposition}\label{characterization_Aboundedness}
  $q \in \crpq{\SSF,\aStar}$ is $\bddIn{A}$ if, and only if, $q$ is "equivalent" to $\qAm{A,m}$ for some $m \in \Nat$. 
\end{proposition}
\begin{proof}
 $(\Leftarrow)$ follows by definition of $\bddIn{A}$. For $(\Rightarrow)$, if $q$ is $\bddIn{A}$ then let $q'$ be the query "equivalent" to $q$ where $q'$ is $\bddIn{A}$ where every "atom" involving $a\in A$ is an \SSF or of the form $a^\le k$.   Let $m$ be the maximum such that $a^{\le m}$ appears in $q'$. 
 
 We claim that $q$ is "equivalent" to $\qAm{A,m}$. Clearly $\qAm{A,m}\semsubset q$.
  Now let $\lambda \in \Exp(q)$. So there is some $\lambda'\in \Exp(q')$ such that there is a "homomorphism" $h:\lambda'\homto \lambda$. Consider the image of $h(\lambda')$. We claim that this forms an "expansion" $\hat \lambda \in \qAm{A,m}$ because all "atoms" that contains $a\in A$ in $\lambda'$ are of the form $a^{\le k}$ for some $K\le m$. Hence, the identity function from $\hat\lambda$ to $\lambda$ gives the required "homomorphism".
\end{proof}

In fact, there is a unique maximal $A \subseteq \alphabet$ such that $q$ is $\bddIn{A}$ because of the following property.

\

\begin{theorem}\label{thm-A-bdd-confluence}
    If a query $q \in \ucrpq{\aSingleton,\aStar}$ is $\bddIn{A}$ and $\bddIn{B}$ then $q$ is also $\bddIn{(A\cup B)}$.
\end{theorem}
\begin{proof}[Proof idea]
\AP
For any $A \subseteq \alphabet$, $n \in \Nat$ and "UCRPQ"$(\aSingleton,\aStar)$ query $q$, let $q[A \mapsto n]$ be the result of replacing in $q$ every $\xrightarrow{a^*}$ with $\xrightarrow{a^{\leq n}}$ for every $a \in A$.  
Hence, by \Cref{characterization_Aboundedness}, $q$ is $\bddIn{A}$  if it is "equivalent" to $q[A \mapsto n]$ for some $n$. 
Let $q[A \mapsto n, B \mapsto m]$ denote $(q[A \mapsto n]) [B \mapsto m]$.

From \Cref{characterization_Aboundedness} let $q$ be equivalent to both $\qAm{A,n}$ and $\qAm{B,m}$ and let $N_q = \max(n,m)$. Since $q$ is $\bddIn{A}$ and $\bddIn{B}$, we have $q[A \mapsto N_q] \semequiv q[B \mapsto N_q] \semequiv q$.
It suffices to show that $q$ is "contained" in $q[A \mapsto N_q, B \mapsto N_q \cdot \nratoms q]$. We will show that for every "expansion" $\lambda$ of $q$ there is an "expansion" $\lambda'$ of $q[A \mapsto N_q, B \mapsto N_q \cdot \nratoms q]$ that maps "homomorphically" to $\lambda$.

Take then such $\lambda$. Since $\lambda$ is "contained" in $q[B \mapsto N_q]$, there is an "expansion" $\lambda_B$ of $q[B \mapsto N_q]$ which maps to $\lambda$. 
Since $q[B \mapsto N_q]$ is "contained" in $q[A \mapsto N_q]$, there is an "expansion" $\lambda_A$ of $q[A \mapsto N_q]$ which maps "homomorphically" to $\lambda_B$. Pick such a $\lambda_A$ of minimal "size".
Hence, we have the "homomorphisms" $h_1: \lambda_A \homto \lambda_B$  and $h_2: \lambda_B \homto \lambda$.

We next show that $\lambda_A$ is in fact an "expansion" of $q[A \mapsto N_q, B \mapsto N_q \cdot \nratoms q]$, which would already prove the statement. By means of contradiction, if it is not the case, there is some $B$-atom "expansion" of length $m > N_q \cdot \nratoms q$ in $\lambda_A$.
 
 This means that the $h_1$-image of such a $B$-atom "expansion" induces a cycle in $\lambda_B$. We can then ``cut out'' the cycle producing another "expansion" $\lambda'_A$ of $q[A \mapsto N_q]$ which is strictly smaller than $\lambda_A$ and which still maps "homomorphically" to $\lambda_B$ via $h_1$ (restricted to $\lambda'_A$). This contradicts the minimality of $\lambda_A$.
\end{proof}

Towards proving the upper bound of \Cref{thm-Max-bdd-result}, \Cref{lem:new:q_equiv_q(n)} and \Cref{lem:new:generating_algo_pi2p} can be generalized trivially:

\begin{lemma}\label{lem:new:q_equiv_q(A,n)}
 Let $q \in \ucrpq{\aSingleton,\aStar}$ and $A\subseteq \alphabet$. If $q$ is $\bddIn{A}$, then it is "equivalent" to $\qAm{A,\Zbound}$ for $\reintro*\Zbound \defeq \nratoms[3]{q} \cdot \Nbound \cdot \nrvars q \cdot \Pi_{w \in \Recwords} |w|$.
\end{lemma}

\begin{lemma}\label{lem:new:generating_algo_pi2p(A,n)}
   Let $q \in \ucrpq{\aSingleton,\aStar}$ and $A\subseteq \alphabet$. Then $q$ is $\bddIn{A}$ "iff" $\qAm{A,\Zbound}$ is "equivalent" to $\qAm{A,\Zboundplus}$, for $\reintro*\Zboundplus \defeq \nratoms q \cdot \Zbound + 1$.
\end{lemma}

The proofs of both statements follow along the same lines as those of \Cref{lem:new:q_equiv_q(n)} and \Cref{lem:new:generating_algo_pi2p} respectively, except that we need to work with $\qAm{A,m}$ "expansions" instead of $\qm m$. Note that in the proof, the operations of extending and contracting of paths in the "recursive" "atom expansions" are done only for 
only those $a^*$ "atoms" such that $a\in A$. For all $b\not\in A$ we keep the "atom expansion" paths unaltered and hence the arguments continue to hold. Together, the two statements above give us the following upper bound.

\begin{theorem}\label{thm-pip2-upper-bound-qAm}
  The $\bddInPb{A}$ problem for $\ucrpq{\aSingleton,\aStar}$ is in "PiP2". Further, an equivalent $\ucrpq{\aSingleton,\aStar[\bar A]}$ query of linear size can be computed, with $\bar A = \alphabet \setminus A$.
 \end{theorem}

 \begin{proof}[Proof of \Cref{thm-Max-bdd-result}]
 \proofcase{Item \eqref{thm-Max-bdd-result:1}}
 The upper bound follows from \Cref{thm-pip2-upper-bound-qAm} and lower bound from \Cref{thm:lowerbound}. 
 
 \proofcase{Item \eqref{thm-Max-bdd-result:2}}
 The existence of a unique maximal $A$ follows from \Cref{thm-A-bdd-confluence}. 
 In order to find it, a "PiP2" algorithm can check if $q$ is $\bddIn{\set a}$ for every $a\in \alphabet$, and we know, thanks to \Cref{thm-A-bdd-confluence}, that the maximal $A$ will be the set of all the $a$'s for which $q$ turned out to be $\bddIn{\set a}$. The "equivalent" $A$-bounded query is then $\qAm{A,\Zbound}$.
\end{proof}

%% file: conclusion.tex
We have shown that the "boundedness problem" for "UCRPQs" with simple recursion of the form $a^*$ is "PiP2"-complete. This is in line  with the complexity for the "containment problem" of similar fragments \cite{FigueiraGKMNT20}.

\smallskip 

\textbf{Constants and free variables.}
While we have focused our study on "Boolean" queries without constants, our upper bounds also extend to queries containing constants and free variables via the classical notion of homomorphism between conjunctive queries with constants and free variables.
For the case of free variables, one can use a standard reduction to Boolean queries:
For any formula with free variables, consider replacing each free variable $x$ with a bound variable $y_x$ and adding a self-loop $y_x \xrightarrow{a_x} y_x$, where $a_x$ is a fresh alphabet symbol depending on $x$. It follows that the original query is "bounded" "iff" the modified Boolean query is "bounded". 

\smallskip 

\textbf{Less trivial recursion.}
Instead of allowing just for one word or one letter to appear under a Kleene star, as in $a^*$, one could also consider the case of a set of letters appearing under a Kleene star, as in $(a + b + c)^*$.
The simplest class of regular expressions containing such behavior, namely the one containing only the regular expression $a$ for each $a \in \alphabet$ and $A^*$ for any set $A \subseteq \alphabet$, is denoted by 
\AP
\knowledgenewrobustcmd{\Astar}{\cmdkl{\mathbf{A}^{\!*}}}
$(\aSingleton,\intro*\Astar)$ in \cite{FigueiraGKMNT20}. The "containment problem" for $\crpq{\aSingleton,\Astar}$ does not enjoy a better complexity than the "containment problem" for general "CRPQ", that is, it is "ExpSpace"-complete.

\AP
We would expect a similar behaviour for the "boundedness problem" of $\ucrpq{\SSF,\Astar}$ but the complexity checking if such a query is "bounded" is open. The "ExpSpace" upper bound follows from \Cref{thm-earlier-bdd-result}, but no matching lower bound is known. We can also generalize this question to \ucrpq{\SSF, \WStar} where $\intro*\WStar$ denotes the regular expressions of the form $(w_1+\ldots +w_k)^*$.
\begin{open}
    What is the complexity for the "boundedness problem" for $\ucrpq{\SSF,\Astar}$ and $\ucrpq{\SSF,\WStar}$?
\end{open}

\smallskip

\textbf{Rewritability of ontology-mediated CRPQs.}
\knowledgenewrobustcmd{\horndl}{\mathcal{ELHI}_{\bot}}
CRPQs can also be considered in the presence of an ontology (see, "eg", \cite{BagetBMT17,Gutierrez-Basulto18}). In such a context, an ontology-mediated query (OMQ) $(T,q)$ admits a rewriting in a language $\+L$ if there is a query $q' \in \+L$ such that for every ABox $A$ we have that $A,T$ entails $q$ if, and only if, $A \models q'$.
For Horn description logics such as $\horndl$, the resulting OMQs $(T,q)$ (where $T \in \horndl$ and $q \in \textup{"CRPQ"}$) is closed under homomorphisms. In these cases the OMQ is FO-rewritable if{f} it is UCQ-rewritable. Hence, our positive complexity result can be seen as a first step towards investigating the rewritability of OMQs with "CRPQs" as query langauges.